\documentclass[aps,nopacs,nokeys,superscriptaddress,11pt,twoside,notitlepage]{revtex4-1}

\usepackage{graphicx,epic,eepic,epsfig,amsmath,latexsym,amssymb,verbatim,color}
\usepackage{tikz}

\usepackage{theorem}
\newtheorem{definition}{Definition}

\newtheorem{lemma}[definition]{Lemma}

\newtheorem{theorem}[definition]{Theorem}
\newtheorem{corollary}[definition]{Corollary}

\def\squareforqed{\hbox{\rlap{$\sqcap$}$\sqcup$}}
\def\qed{\ifmmode\squareforqed\else{\unskip\nobreak\hfil
\penalty50\hskip1em\null\nobreak\hfil\squareforqed
\parfillskip=0pt\finalhyphendemerits=0\endgraf}\fi}
\def\endenv{\ifmmode\;\else{\unskip\nobreak\hfil
\penalty50\hskip1em\null\nobreak\hfil\;
\parfillskip=0pt\finalhyphendemerits=0\endgraf}\fi}
\newenvironment{proof}{\noindent \textbf{{Proof~} }}{\qed}
\newenvironment{proof-of}[1]{\noindent \textbf{{Proof~#1} }}{\qed}
\newenvironment{remark}{\noindent \textbf{{Remark~}}}{\qed}
\newenvironment{example}{\noindent \textbf{{Example~}}}{\qed}

\mathchardef\ordinarycolon\mathcode`\:
\mathcode`\:=\string"8000
\def\vcentcolon{\mathrel{\mathop\ordinarycolon}}
\begingroup \catcode`\:=\active
  \lowercase{\endgroup
  \let :\vcentcolon
  }

\newcommand{\nc}{\newcommand}
\nc{\rnc}{\renewcommand}
\nc{\beq}{\begin{equation}}
\nc{\eeq}{{\end{equation}}}
\nc{\beqa}{\begin{eqnarray}}
\nc{\eeqa}{\end{eqnarray}}
\nc{\lbar}[1]{\overline{#1}}
\nc{\bra}[1]{\langle#1|}
\nc{\ket}[1]{|#1\rangle}
\nc{\ketbra}[2]{|#1\rangle\!\langle#2|}
\nc{\braket}[2]{\langle#1|#2\rangle}
\nc{\proj}[1]{| #1\rangle\!\langle #1 |}
\nc{\avg}[1]{\langle#1\rangle}
\nc{\Rank}{\operatorname{Rank}}
\nc{\smfrac}[2]{\mbox{$\frac{#1}{#2}$}}
\nc{\tr}{\operatorname{Tr}}
\nc{\ox}{\otimes}
\nc{\dg}{\dagger}
\nc{\dn}{\downarrow}
\nc{\cA}{{\cal A}}
\nc{\cB}{{\cal B}}
\nc{\cC}{{\cal C}}
\nc{\cD}{{\cal D}}
\nc{\cE}{{\cal E}}
\nc{\cF}{{\cal F}}
\nc{\cG}{{\cal G}}
\nc{\cH}{{\cal H}}
\nc{\cI}{{\cal I}}
\nc{\cJ}{{\cal J}}
\nc{\cK}{{\cal K}}
\nc{\cL}{{\cal L}}
\nc{\cM}{{\cal M}}
\nc{\cN}{{\cal N}}
\nc{\cO}{{\cal O}}
\nc{\cP}{{\cal P}}
\nc{\cQ}{{\cal Q}}
\nc{\cR}{{\cal R}}
\nc{\cS}{{\cal S}}
\nc{\cT}{{\cal T}}
\nc{\cX}{{\cal X}}
\nc{\cY}{{\cal Y}}
\nc{\cZ}{{\cal Z}}
\nc{\csupp}{{\operatorname{csupp}}}
\nc{\qsupp}{{\operatorname{qsupp}}}
\nc{\var}{{\operatorname{var}}}
\nc{\rar}{\rightarrow}
\nc{\lrar}{\longrightarrow}
\nc{\polylog}{{\operatorname{polylog}}}
\nc{\1}{{\openone}}
\nc{\wt}{{\operatorname{wt}}}
\nc{\av}[1]{{\left\langle {#1} \right\rangle}}

\nc{\RR}{{{\mathbb R}}}
\nc{\CC}{{{\mathbb C}}}
\nc{\FF}{{{\mathbb F}}}
\nc{\NN}{{{\mathbb N}}}
\nc{\ZZ}{{{\mathbb Z}}}
\nc{\PP}{{{\mathbb P}}}
\nc{\QQ}{{{\mathbb Q}}}
\nc{\UU}{{{\mathbb U}}}
\nc{\EE}{{{\mathbb E}}}
\nc{\id}{{\operatorname{id}}}

\nc{\CHSH}{{\operatorname{CHSH}}}

\nc{\be}{\begin{equation}}
\nc{\ee}{{\end{equation}}}
\nc{\bea}{\begin{eqnarray}}
\nc{\eea}{\end{eqnarray}}
\nc{\<}{\langle}
\rnc{\>}{\rangle}
\nc{\Hom}[2]{\mbox{Hom}(\CC^{#1},\CC^{#2})}
\nc{\rU}{\mbox{U}}

\nc{\ob}[1]{#1}

\nc{\SEP}{{\text{SEP}}}
\nc{\NS}{{\text{NS}}}
\nc{\SNOS}{{\text{SNOS}}}
\nc{\LOCC}{{\text{LOCC}}}
\nc{\PPT}{{\text{PPT}}}
\nc{\EXT}{{\text{EXT}}}
\nc{\Sym}{{\operatorname{Sym}}}

\nc{\ERLO}{{E_{\text{r,LO}}}}
\nc{\ERLOCC}{{E_{\text{r,LOCC}}}}
\nc{\ERPPT}{{E_{\text{r,PPT}}}}
\nc{\ERLOCCinfty}{{E^{\infty}_{\text{r,LOCC}}}}
\nc{\Aram}{{\operatorname{\sf A}}}

\DeclareFontFamily{U}{mathx}{\hyphenchar\font45}
\DeclareFontShape{U}{mathx}{m}{n}{
      <5> <6> <7> <8> <9> <10>
      <10.95> <12> <14.4> <17.28> <20.74> <24.88>
      mathx10
      }{}
\DeclareSymbolFont{mathx}{U}{mathx}{m}{n}
\DeclareMathSymbol{\bigtimes}{1}{mathx}{"91}

\begin{document}

\title{Parallel repetition and concentration for (sub-)no-signalling games \protect\\
       via a flexible constrained de Finetti reduction}

\author{C\'{e}cilia Lancien}
\email{lancien@math.univ-lyon1.fr}
\affiliation{Departament de F\'{\i}sica: Grup d'Informaci\'{o} Qu\`{a}ntica, Universitat Aut\`{o}noma de Barcelona, 08193 Bellaterra, Barcelona, Spain}
\affiliation{Institut Camille Jordan, Universit\'{e} Claude Bernard Lyon 1, 43 boulevard du 11 novembre 1918, 69622 Villeurbanne Cedex, France}

\author{Andreas Winter}
\email{andreas.winter@uab.cat}
\affiliation{Departament de F\'{\i}sica: Grup d'Informaci\'{o} Qu\`{a}ntica, Universitat Aut\`{o}noma de Barcelona, 08193 Bellaterra, Barcelona, Spain}
\affiliation{ICREA -- Instituci\'{o} Catalana de Recerca i Estudis Avan\c{c}ats, Pg.~Llu\'{\i}s Companys, 23, 08010 Barcelona, Spain}

\begin{abstract}
We use a recently discovered constrained de Finetti reduction (aka ``Post-Selection Lemma'') to study the parallel repetition of multi-player non-local games under no-signalling strategies. Since the technique allows us to reduce general strategies to independent plays, we obtain parallel repetition (corresponding to winning all rounds) in the same way as exponential concentration of the probability to win a fraction larger than the value of the game.

Our proof technique leads us naturally to a relaxation of no-signalling (NS) strategies, which we dub \emph{sub-no-signalling (SNOS)}. While for two players the two concepts coincide, they differ for three or more players. Our results are most complete and satisfying for arbitrary number of sub-no-signalling players, where we get universal parallel repetition and concentration for any game, while the no-signalling case is obtained as a corollary, but only for games with ``full support''.
\end{abstract}

\date{3rd July 2016}

\maketitle

\thispagestyle{empty}

\section{Non-local multi-player games and their parallel repetition: short review of previous approaches and ours}
\label{sec:games}

A multi-player non-local game is played between cooperating but non-communicating
players. Each player receives an input from some input alphabet and has to
produce an output in some output alphabet. The common goal of the players
is to satisfy some pre-defined predicate on their inputs and outputs. For that,
they may agree on a strategy before the game starts, but are then not allowed to
communicate anymore. Such games are especially relevant in theoretical physics in
the context of the foundations of quantum mechanics and quantum information,
and in computer science where they arise in multi-prover interactive proof
systems.
Indeed, they may provide an intuitive and quantitative understanding of the role
played by various degrees of correlations in global systems which are composed of
several local subsystems. These games also arise in complexity theory, under the
formulation of multi-provers with some shared resources producing a protocol that
should convince a referee, or in cryptography as attacks from malicious parties
having a more or less restricted physical power.

The \emph{value} of a game is the maximum winning probability of the players,
over all allowed joint strategies, using possibly some prescribed correlation
resource such as shared randomness, quantum entanglement or no-signalling
correlations. It has been a subject of considerable study how the availability of different
resources affects the values of certain games \cite{Bell,C-H-S-H,Tsi,P-R,B-L-M-P-P-R}.

In this context, a natural question is how the value of a game behaves when $n$ independent
instances of the game are played simultaneously, i.e.~each player gets
$n$ independent inputs and has to provide $n$ outputs such that each
game instance is won (or a large fraction of them). This is the parallel
repetition problem. Playing independently the optimal single-game strategy
on all $n$ game instances will result in an exponentially decreasing
winning probability. But although that was found paradoxical at first, this
is in general not optimal \cite{Feige,F-V}.
For classical two-player games, Raz~\cite{Raz}, later simplified and
improved by Holenstein~\cite{Hol}, established the first general parallel
repetition theorem, showing that the value of $n$ repetitions decreases
exponentially for every game. Holenstein~\cite{Hol} also proved an analogous
parallel repetition theorem for the no-signalling value of general
two-player games. Only recently, parallel repetition theorems
were proved for the entangled value of two-player games: for general games, nothing better than a polynomial decay result is known up to now (this was intially proved for a slightly modified game \cite{K-V}, and very recently only for the game itself \cite{Yuen}), while exponential decay results have been established in several special cases (perfect parallel repetition for XOR games~\cite{C-S-U-U}, exponential decrease under parallel repetition for unique games~\cite{K-R-T}, projection games~\cite{D-S-V}, free games~\cite{C-S,J-P-Y}).

Even less is known concerning multi-player games. And apart from \cite{C-W-Y} (containing both classical and quantum statements), results were obtained only in the no-signalling setting~\cite{B-F-S,AF-R-V,Ros}.
The present work has the same focus on multiple no-signalling players,
albeit we will find that the theory becomes much more satisfying
for \emph{sub-no-signalling} players.

Before getting into more precise and technical statements, let us give a high-level exposition of the philosophy of the present work, and especially how it compares to or differs from previous approaches. The standard proof technique to tackle parallel repetition (in either the classical, the quantum or the no-signalling case) consists in iteratively assuming that the players have won a given instance of the game and then studying how this affects their winning probability in the others. Hence, if one can show that, conditioned on the event ``the players have already won $k$ instances of the game'', the probability is high that they lose in at least $1$, resp.~most, of the $n-k$ remaining instances, one gets exponential decay of the probability of winning all, resp.~a fraction above the game value of, $n$ instances of the game played in parallel. The main drawback of this approach is probably its ``locality'', which makes it not so straightforward and not so easily generalizable to more than two players. Here, we take a more ``global'' look at the problem, by attempting to reduce the study of such $n$-instance game to that of $n$ i.i.d.~$1$-instance games, whose analysis is trivial. This is where de Finetti type statements come into play: using the fact that the repeated game is symmetric under permutation of its parallel rounds, these allow relating it in some way to independent rounds.

However, there are certain steps from the standard route which we do not avoid in our approach. One of them is some kind of reconstruction step. Phrased informally: we have to be able to say at some point that, if our strategy almost satisfies the constraints defining our set of interest, then there must exist a strategy which exactly satisfies them and which is not too far away from it. Nonetheless, it is not always so easy to get handleable quantitative versions of this quite natural expectancy. This is the main reason why the set of sub-no-signalling strategies that we introduce is such a nice one for studying parallel repetition. Indeed, we prove that if a strategy satisfies all the sub-no-signalling constraints, up to some error $\epsilon$, then it is $C\epsilon$-close to the set of sub-no-signalling strategies, with a constant $C$ which depends only on the number of players. And this fact ultimately translates into a universal exponential decay statement for the sub-no-signalling value of repeated multi-player games. This kind of stability property actually also holds for the set of two-player no-signalling strategies, which was discovered and used by Holenstein to prove universal parallel repetition in that case \cite{Hol}. Oppositely, it remains unknown whether this is still true for three or more players, which explains why all parallel repetition results for strictly more than two no-signalling players are game-dependent ones \cite{B-F-S,AF-R-V}. Viewing the no-signalling setting inside our broader sub-no-signalling framework, we are also able to reproduce these earlier findings. One notable advantage of our approach is that it is particularly well-suited to studying the concentration problem as well, and once exponential decay of the probability of winning all game instances is established, exponential decay of winning a too high fraction of them comes almost for free.

\section{Non-local multi-player games and (sub-)no-signalling strategies: definitions and first observations}

Specifically, we will consider here $\ell$-player games $G$ with input alphabets
$\cX_1,\ldots,\cX_\ell$ and output alphabets $\cA_1,\ldots,\cA_\ell$.
By way of notation,
\[ \underline{\cX} := \bigtimes_{i=1}^\ell \cX_i\ \text{and}\ \underline{\cA} := \bigtimes_{i=1}^\ell \cA_i. \]
Furthermore,
for any subset $I\subset [\ell]$ of indices,
\[ \cX_I := \bigtimes_{i\in I}\cX_i\ \text{and}\ \cA_I := \bigtimes_{i\in I} \cA_i. \]
An element from $\cX_i,\cX_I,\underline{\cX}$ will usually be denoted by $x_i,x_I,\underline{x}$, respectively, sometimes without explicitly specifying the set it belongs to (and similarly for $\cA_i,\cA_I,\underline{\cA}$).

Also, for any $I,J\subset [\ell]$, given $T$ a probability distribution (which we may quite often abbreviate by ``p.d.'') on $\cX_I$, resp.~$P$ a conditional probability distribution on $\cA_J|\cX_I$, we may denote it by $T_{\cX_I}$, resp.~$P_{\cA_J|\cX_I}$, when confusion on the considered alphabets is at risk.

From now on, we will be interested in making minimal a priori assumptions on how powerful
the $\ell$ players may be. This will naturally lead us to considering that
their common strategy to win the game $G$ could be any no-signalling (or even
sub-no-signalling) strategy, which we define now.

\begin{definition} \label{def:NS-SNOS}
  The sets of \emph{no-signalling} and \emph{sub-no-signalling}
  correlations, denoted respectively $\NS(\underline{\cA}|\underline{\cX})$
  and $\SNOS(\underline{\cA}|\underline{\cX})$,
  consist of non-negative densities $P(\underline{a}|\underline{x}) \geq 0$
  defined as follows:
  \begin{equation}
    \label{eq:NS}
    P \in \NS(\underline{\cA}|\underline{\cX})
      :\Leftrightarrow
      \forall\ I\subsetneq[\ell],\ \exists\ Q(\cdot|x_I)\text{ p.d.'s on }\cA_I\text{ s.t.~}
      \forall\ \underline{x},a_I,\ P(a_I|\underline{x}) = Q(a_I|x_I),
  \end{equation}
  \begin{equation}
    \label{eq:SNOS}
    P \in \SNOS(\underline{\cA}|\underline{\cX})
      :\Leftrightarrow
      \forall\ I\subsetneq[\ell],\ \exists\ Q(\cdot|x_I)\text{ p.d.'s on }\cA_I\text{ s.t.~}
      \forall\ \underline{x},a_I,\ P(a_I|\underline{x}) \leq Q(a_I|x_I).
  \end{equation}
  Here, $P(a_I|\underline{x})$ denotes the marginal density,
  \[
    P(a_I|\underline{x}) = \sum_{a_{I^c} \in \cA_{I^c}} P(\underline{a}=a_I a_{I^c}|\underline{x}).
  \]
\end{definition}

\begin{remark}
Note that under this definition,
$\NS(\underline{\cA}|\underline{\cX}) \subset \SNOS(\underline{\cA}|\underline{\cX})$,
but the latter is a strictly larger set (e.g.~it always contains the
all-zero density). Furthermore,
$P \in \NS(\underline{\cA}|\underline{\cX})$
iff $P \in \SNOS(\underline{\cA}|\underline{\cX})$ and
$P$ is \emph{normalized} in the sense that for all
$\underline{x} \in \underline{\cX}$,
$\sum_{\underline{a}} P(\underline{a}|\underline{x}) = 1$.
Indeed, $\NS$ consists of conditional probability distributions, while $\SNOS$
allows, given each input, a total ``probability'' of less than or equal to $1$.

Also, it can be shown that in equation~(\ref{eq:NS}), only sets of the
form $I=[\ell]\setminus i$ need to be considered. This is because
the no-signalling conditions take the form of equations and
this subset spans the set of all equations required (cf.~\cite{Han}, Lemma 2.7). The analogous
statement for sub-no-signalling is not known and likely false.
Nevertheless, one might in other contexts consider to relax
the conditions of equation~(\ref{eq:SNOS}) to hold only for a selected family
of subsets $I \subset [\ell]$.
\end{remark}

An $\ell$-player game $G$ is characterized by a probability distribution
$T(\underline{x})$ on the queries $\underline{\cX}$, and a binary
predicate $V(\underline{a},\underline{x}) \in \{0,1\}$
on the answers and queries $\underline{\cA}\times\underline{\cX}$, as illustrated in Figure \ref{fig:game}. The no-signalling, resp.~ sub-no-signalling, value of the game, denoted $\omega_{\NS}(G)$, resp.~$\omega_{\SNOS}(G)$, is the maximum of the winning probability
\[ \PP\left(\text{win}\right) = \EE\, V(\underline{A},\underline{X})
                    = \sum_{\underline{a},\underline{x}} T(\underline{x}) V(\underline{a},\underline{x})
                                                                          P(\underline{a}|\underline{x}) \]
over all $P\in\NS(\underline{\cA}|\underline{\cX})$, resp.~$P\in \SNOS(\underline{\cA}|\underline{\cX})$, where the distribution of
$\underline{X}=X_1\ldots X_\ell$ and $\underline{A}=A_1\ldots A_\ell$
is as expected,
\[ \forall\ \underline{x},\underline{a},\ \PP\left(\underline{X}=\underline{x},\, \underline{A}=\underline{a}\right)
                           = T(\underline{x}) P(\underline{a}|\underline{x}). \]

In words, the (sub-)no-signalling value of a game is the maximal probability of winning it when no limitation is assumed on the power of the players, apart from the fact that they cannot signal information instantaneously from one another. In the sub-no-signalling case, constraints are relaxed even more: players are not forced to always produce an output, and it is only required that their strategy ``looks as if it were no-signalling'' (even though they may have ``hidden'' in their abstentions the fact that it is signalling). In Section \ref{sec:end}, we extend on the physical interpretation of sub-no-signalling, and briefly discuss other kinds of restrictions that one may put on the players' physical power, such as shared randomness or shared quantum entanglement only.

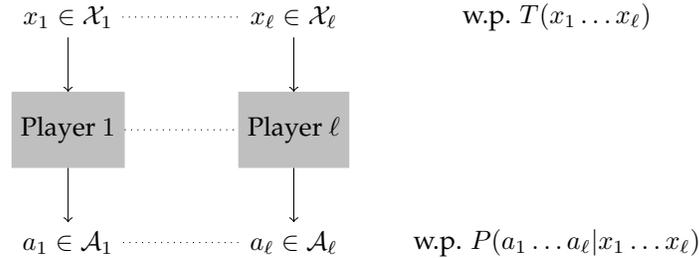
\begin{figure}[h]
\caption{An $\ell$-player non local game}
\begin{center}
\begin{tikzpicture} [scale=1]
\node[draw=lightgray, minimum height=1cm, minimum width=1cm, fill=lightgray] (X) at (0,0) {Player $1$};
\node[draw=white, minimum height=0.5cm, minimum width=1cm, fill=white] (X') at (0,1.5) {$x_1\in\mathcal{X}_1$};
\node[draw=white, minimum height=0.5cm, minimum width=1cm, fill=white] (X'') at (0,-1.5) {$a_1\in\mathcal{A}_1$};
\node[draw=lightgray, minimum height=1cm, minimum width=1cm, fill=lightgray] (Y) at (3,0) {Player $\ell$};
\node[draw=white, minimum height=0.5cm, minimum width=1cm, fill=white] (Y') at (3,1.5) {$x_\ell\in\mathcal{X}_\ell$};
\node[draw=white, minimum height=0.5cm, minimum width=1cm, fill=white] (Y'') at (3,-1.5) {$a_\ell\in\mathcal{A}_\ell$};
\node[draw=white, minimum height=0.5cm, minimum width=1cm, fill=white] at (6.5,1.5) {w.p. $T(x_1\ldots x_\ell)$};
\node[draw=white, minimum height=0.5cm, minimum width=1cm, fill=white] at (6.5,-1.5) {w.p. $P(a_1\ldots a_\ell|x_1\ldots x_\ell)$};
\node[draw=white, minimum height=0.5cm, minimum width=6cm, fill=white] at (4,-2.5) {The players win iff $V(a_1\ldots a_\ell,x_1\ldots x_\ell)=1$};
\draw [->] (X'.south) -- (X.north); \draw [->] (X.south) -- (X''.north); \draw [->] (Y'.south) -- (Y.north); \draw [->] (Y.south) -- (Y''.north);
\draw [dotted] (X.east) -- (Y.west); \draw [dotted] (X'.east) -- (Y'.west); \draw [dotted] (X''.east) -- (Y''.west);
\end{tikzpicture}
\end{center}
\label{fig:game}
\end{figure}

\subsection{Two-player SNOS $\mathbf{\equiv}$ NS}

Not surprisingly, the no-signalling and sub-no-signalling values of
games are related. We start by showing that for any two-player game $G$,
they are identical, i.e.~$\omega_{\NS}(G) = \omega_{\SNOS}(G)$.
As $\NS \subset \SNOS$, the
inequality ``$\leq$'' is evident, and we only need to prove the opposite
inequality ``$\geq$''. This follows from the following structural lemma.

\begin{lemma}[cf.~\cite{Ito}]
  \label{lemma:2-player-SNOS}
  Let $P \in \SNOS(\cA\times\cB|\cX\times\cY)$ be a two-player
  sub-no-signalling correlation. Then there exists a no-signalling
  correlation $P' \in \NS(\cA\times\cB|\cX\times\cY)$ with $P\leq P'$
  pointwise, i.e.~$P(ab|xy) \leq P'(ab|xy)$ for all $a,b,x,y$.
\end{lemma}

Since playing a game $G$ with a strategy $P$ necessarily yields a smaller value than playing it with a strategy $P'$ which dominates $P$ pointwise, it is clear that once Lemma \ref{lemma:2-player-SNOS} is proved we just have to apply it to $P$ an optimal $\SNOS$ strategy for $G$ to get the inequality ``$\geq$''.


\begin{proof}
  If $P$ is normalized, i.e.~if for all $x,y$,
  $\sum_{ab} P(ab|xy)=1$, there is nothing to prove because
  $P$ is already no-signalling.

  Otherwise, there exist $x,y$ with weight
  $\sum_{ab} P(ab|xy) =: w < 1$. By sub-no-signalling assumption, we have
  distributions $Q(a|x)$ and $Q(b|y)$ dominating the marginals:
  \[
    \forall\ a,b,\ P(a|xy) \leq Q(a|x),\ P(b|xy) \leq Q(b|y).
  \]
  As the total weight of both marginals of $P(\cdot|xy)$ is $w < 1$, we can find
  $a$ and $b$ such that
  \[
    P(a|xy) < Q(a|x), \quad P(b|xy) < Q(b|y),
  \]
  so we can increase $P(ab|xy)$ by some $\epsilon > 0$
  to $P'(ab|xy) = P(ab|xy) + \epsilon$ and still satisfy the
  sub-no-signalling conditions. By choosing $\epsilon$ maximally
  so, we can reduce the total number of strict inequality signs
  in the SNOS conditions. Iterating this procedure we arrive
  at a sub-no-signalling correlation $P'$ with all inequalities
  met with equality, i.e.~a no-signalling correlation.

  Another presentation of this argument appeals to compactness.
  Consider the following set of correlations:
  \[ \mathrm{X}_{P,Q} :=\left\{ P'\ :\ \forall\ a,b,x,y,\ P'(ab|xy)\geq P(ab|xy),\ P'(a|x)\leq Q(a|x),\ P'(b|y)\leq Q(b|y) \right\}. \]
  $\mathrm{X}_{P,Q}$ being compact and $P'\mapsto \sum_{xy} \sum_{ab} P'(ab|xy)$ being continuous,
  \[ \sup \left\{ \sum_{xy} \sum_{ab} P'(ab|xy)\ :\ P'\in\mathrm{X}_{P,Q} \right\} \]
  is actually attained. If it were less than $|\cX\times\cY|$, we could use the procedure
  above to increase the objective function, contradicting that it
  is a maximum.
\end{proof}

Note that the ``bumping up'' procedure described above, in order to transform any two-player sub-no-signalling strategy into a no-signalling one dominating it pointwise, may fail for more players. The two-player case is indeed special, due to non-overlapping of the two SNOS or NS constraints. However, already in the case of three players, even just the three inequalities $P_{\cA_i\cA_j|\cX_i\cX_j\cX_k}\leq Q_{\cA_i\cA_j|\cX_i\cX_j}$ may be impossible to bring simultaneously to equalities by pointwise increment (as illustrated by Example \ref{ex} below).

\subsection{Multi-player SNOS vs NS}

Clearly, $\omega_{\NS}(G) \leq \omega_{\SNOS}(G)$ for every game,
and there are examples of games (with game distribution $T$ having
strictly smaller than full support) where $\omega_{\NS}(G) < 1$
but $\omega_{\SNOS}(G)=1$, for instance the \textit{anticorrelation
game}.

\begin{example} \label{ex}
\normalfont
{\textbf (cf.~\cite{AF-R-V}, Appendix A)}
Consider the three-player \emph{anti-correlation game} $A_3$, which
has binary input and output for all players and game
distribution $T$ supported on $\{0,1\}^3\setminus\{111\}$,
i.e.~$111$ does not occur as a triple of questions.
The winning predicate is that if any two inputs are $1$,
say $x_i=x_j=1$, then the corresponding outputs must be
different, $a_i \neq a_j$. While if there are zero or only a single
$1$ amongst the inputs, outputs may be arbitrary.

It is straightforward to verify that the following correlation
is in $\SNOS\left(\{0,1\}^3|\{0,1\}^3\right)$ and wins the game with certainty:
\[
  P(a_1a_2a_3|x_1x_2x_3) = \begin{cases}
                             0                         & \text{ if } x_1x_2x_3=111, \\
                             1/8                   & \text{ if } \exists\ 1\leq i\neq j\leq 3:\ x_i=x_j=0, \\
                             \delta_{a_i,1-a_j}/4 & \text{ if } \exists\ 1\leq i\neq j\leq 3:\ x_i=x_j=1 \text{ and } x_1x_2x_3 \neq 111.
                           \end{cases}
\]
So $\omega_{\SNOS}(A_3)=1$. On the other hand, for, say, $T$ uniform on $\{011,101,110\}$,
one can check by elementary means that $\omega_{\NS}(A_3) = 2/3$.
\end{example}

What happens in the above example is that it is possible to satisfy any two amongst the three no-signalling constraints, but not the three of them at the same time. This is a phenomenon sometimes referred to as ``frustration''.

However, for a game distribution $T$ having full support, a simple reasoning shows that
$\omega_{\NS}(G) < 1$ implies $\omega_{\SNOS}(G) < 1$. Indeed,
we show the contrapositive, assuming that $\omega_{\SNOS}(G) = 1$.
Because of the full support of $T$, this implies that for the
optimal sub-no-signalling strategy $P$ and every $\underline{x}$,
\[
  1 =    \sum_{\underline{a}} V(\underline{a},\underline{x}) P(\underline{a}|\underline{x})
    \leq \sum_{\underline{a}} P(\underline{a}|\underline{x})
    \leq 1,
\]
hence equality (i.e.~normalization) holds for all $\underline{x}$.
Thus, $P$ is really a no-signalling correlation and so
$\omega_{\NS}(G)=1$.
In fact, we can show something stronger, namely the following quantitative relationship.

\begin{lemma}
  \label{lemma:SNOS-NS-multiplayers}
  Consider a game distribution $T$ with full support on $\underline{\cX}$.
  Then there exists $\Gamma=\Gamma(T)\geq 0$, which only depends
  on $T$, such that for every game $G$ with query distribution
  $T$,
  \[ \omega_{\SNOS}(G) \geq 1-\epsilon \ \Rightarrow\ \omega_{\NS}(G) \geq 1-(\Gamma+1)\epsilon. \]
  The definition of $\Gamma$ can be taken from~\cite{B-F-S} or~\cite{AF-R-V}, where it is implicitly defined as some robustness parameter of the linear program whose optimal value is $\omega_{\NS}(G)$.
\end{lemma}

\begin{proof}
Take an optimal strategy $P\in\SNOS(\underline{\cA}|\underline{\cX})$,
so that $P(a_I|\underline{x}) \leq Q(a_I|x_I)$ for all $I$, $a_I$, $\underline{x}$. Then,
\[
  \sum_{\underline{a},\underline{x}} T(\underline{x}) P(\underline{a}|\underline{x})
                                                \geq \sum_{\underline{a},\underline{x}} T(\underline{x})V(\underline{a},\underline{x}) P(\underline{a}|\underline{x}) = \omega_{\SNOS}(G) \geq 1-\epsilon.
\]
And so we get, for all $I$,
\[
  \bigl\| T_{\underline{\cX}}P_{\cA_I|\underline{\cX}} - T_{\underline{\cX}}Q_{\cA_I|\cX_I} \bigr\|_1
    =    \sum_{a_I,\underline{x}}  T(\underline{x}) \bigl( Q(a_I|x_I) - P(a_I|\underline{x}) \bigr)
    \leq \epsilon,
\]
because the difference term in the sum is non-negative.

Now simply ``bump up'' the sub-normalized probability distribution
$P_{\underline{\cA}|\underline{\cX}}$ to a
properly normalized conditional probability distribution
$P_{\underline{\cA}|\underline{\cX}}'$, adding at most an averaged weight over $T_{\underline{\cX}}$ of $\epsilon$, and hence, for all $I$,
\[
  \frac12 \bigl\|  T_{\underline{\cX}}P_{\cA_I|\underline{\cX}}' -  T_{\underline{\cX}}Q_{\cA_I|\cX_I} \bigr\|_1 \leq \epsilon.
\]

At this point we can invoke the stability of linear programs,
used in~\cite{B-F-S} and~\cite{AF-R-V} to conclude that there
is $\Gamma = \Gamma(T)\geq 0$ such that there is a no-signalling
correlation $P_{\underline{\cA}|\underline{\cX}}'' \in \NS(\underline{\cA}|\underline{\cX})$ with
\[
  \frac12 \bigl\| T_{\underline{\cX}}P_{\underline{\cA}|\underline{\cX}}''
                   - T_{\underline{\cX}}P_{\underline{\cA}|\underline{\cX}}' \bigr\|_1 \leq \Gamma\epsilon.
\]
This gives
\begin{align*} \omega_{\NS}(G) \geq & \sum_{\underline{a},\underline{x}} T(\underline{x})V(\underline{a},\underline{x}) P''(\underline{a}|\underline{x})\\
\geq & \sum_{\underline{a},\underline{x}} T(\underline{x})V(\underline{a},\underline{x}) P'(\underline{a}|\underline{x}) -\Gamma\epsilon \\
\geq & \sum_{\underline{a},\underline{x}} T(\underline{x})V(\underline{a},\underline{x}) P(\underline{a}|\underline{x}) -\Gamma\epsilon \\
\geq &\, 1-(\Gamma+1)\epsilon,
\end{align*}
where we have used the total variational bound on $P''-P'$, the fact that
$P'$ dominates $P$ and the assumption on the probability of winning $G$ when played $P$.
\end{proof}

The rest of the paper is structured as follows: In Section~\ref{sec:main}
we introduce parallel repetition of games, and state our main
results, which improve upon, and partly clarify, earlier findings
by Holenstein~\cite{Hol}, Buhrman \emph{et al.}~\cite{B-F-S}
and Arnon-Friedman \emph{et al.}~\cite{AF-R-V}.
In Section~\ref{sec:constrained-postselection}, we present the
main technical tool, one of the constrained de Finetti reductions from \cite{L-W},
adapted to our present needs, followed by the proofs of the
main theorems and corollaries in Section~\ref{sec:proofs}. We
conclude in Section~\ref{sec:end}.

\section{Parallel repetition: definitions and main results}
\label{sec:main}

Given an $\ell$-player game $G$, with probability distribution
$T(\underline{x})$ on $\underline{\cX}$ and binary
predicate $V(\underline{a},\underline{x}) \in \{0,1\}$
on $\underline{\cA}\times\underline{\cX}$, we are interested
in playing the same game $n$ times independently in parallel,
and in looking at the probability of winning all $n$ or a subset of $t$ of them.

Formally, the \emph{$n$-fold parallel repetition of $G$}
is the $\ell$-player game $G^n$ having the product probability distribution on $\underline{\cX}^n$
\[ T^{\otimes n}(\underline{x}^n)=T\bigl(\underline{x}^{(1)})\cdots T(\underline{x}^{(n)}\bigr), \]
and the product binary predicate on $\underline{\cA}^n\times\underline{\cX}^n$
\[ V^{\otimes n}(\underline{a}^n,\underline{x}^n)
  =V\bigl(\underline{a}^{(1)},\underline{x}^{(1)})\cdots V(\underline{a}^{(n)},\underline{x}^{(n)}\bigr) \in \{0,1\}. \]
The no-signalling, resp.~sub-no-signalling, value of this $n$-fold parallel repetition game,
denoted $\omega_{\NS}(G^n)$, resp.~$\omega_{\SNOS}(G^n)$, is thus the maximum
of the winning probability
\[ \PP\left(\text{win}\right) = \sum_{\underline{a}^n,\underline{x}^n} T^{\otimes n}(\underline{x}^n) V^{\otimes n}(\underline{a}^n,\underline{x}^n)
                                                                          P(\underline{a}^n|\underline{x}^n) \]
over all $P\in\NS(\underline{\cA}^n|\underline{\cX}^n)$, resp.~$P\in \SNOS(\underline{\cA}^n|\underline{\cX}^n)$.

In words, the players win $G^n$ if they win all $n$ instances of $G$ played
in parallel. So we obviously always have (for the allowed set of strategies
being $X \in \{\NS,\SNOS\}$)
\begin{equation}
  \label{eq:G-G^n}
  \big(\omega_X(G)\big)^n \leq \omega_X(G^n) \leq \omega_X(G).
\end{equation}
However, in the case where $\omega_X(G)<1$, the gap between the lower and
upper bounds in equation~\eqref{eq:G-G^n} grows exponentially with $n$, making
equation \eqref{eq:G-G^n} very little informative. The parallel repetition problem is thus the following:
If none of the players' allowed strategies can make them win $1$ instance
of $G$ with probability $1$, does it necessarily imply that they have an
exponentially decaying probability of winning $n$ of them at the same time?
And if so at which rate?

More generally, we can study the game $G^{t/n}$, whose winning predicate is
defined as winning any $t$ (or more) out of $n$ repetitions \cite{B-F-S}, i.e.~\[ V^{t/n}(\underline{a}^n,\underline{x}^n) := \left\{ \sum_{i=1}^n V\bigl(\underline{a}^{(i)},\underline{x}^{(i)}\bigr) \geq t \right\} = \begin{cases} 1 & \text{ if }\sum_{i=1}^n V\bigl(\underline{a}^{(i)},\underline{x}^{(i)}\bigr) \geq t, \\
0 & \text{ otherwise}. \end{cases} \]
Note that, with our notation, $G^n = G^{n/n}$.

\medskip
The main results of the present paper are gathered below, where we set $C_{\ell}:=2^{\ell+1}-3$.

\begin{theorem}[Parallel repetition of sub-no-signalling $\ell$-player games] \label{th:repetitionSNOS}
Let $G$ be an $\ell$-player game such that
$\omega_{\SNOS}(G)\leq 1-\delta$ for some $0<\delta<1$.
Then, for any $n\in\NN$, and any $t \geq (1-\delta+\alpha)n$ for some $0<\alpha\leq\delta$, we have
\begin{align*}
& \omega_{\SNOS}(G^n) \leq \left(1-\frac{\delta^2}{5C_{\ell}^2}\right)^n,\\
& \omega_{\SNOS}(G^{t/n}) \leq \exp\left( -n \frac{\alpha^2}{5C_{\ell}^2} \right).
\end{align*}
\end{theorem}

As immediate consequences or refinements of Theorem \ref{th:repetitionSNOS}, we can get
parallel repetition results for the no-signalling value of multiplayer games
in some particular instances.

\begin{corollary}[Parallel repetition of no-signalling full support $\ell$-player games, cf.~\cite{B-F-S,AF-R-V}]
\label{cor:repetitionSNOS}
Let $G$ be an $\ell$-player game whose distribution $T$ has full support,
and such that $\omega_{\NS}(G)\leq 1-\delta$ for some $0<\delta<1$.
Then, for any $n\in\NN$, and any $t \geq (1-\delta+\alpha)n$ for some $0<\alpha\leq\delta$, we have
\begin{align*}
& \omega_{\NS}(G^n) \leq \left(1-\frac{\delta^2}{5C_{\ell}^2(\Gamma+1)^2}\right)^n,    \\
& \omega_{\NS}(G^{t/n}) \leq \exp\left( -n \frac{\alpha^2}{5C_{\ell}^2(\Gamma+1)^2} \right),
\end{align*}
where $\Gamma=\Gamma(T)\geq 0$ is the constant from Lemma \ref{lemma:SNOS-NS-multiplayers},
which only depends on $T$.
\end{corollary}

Note that the constant $\Gamma$ in this corollary depends on the game,
and in the worst case carries a heavy dependence on the players' alphabet
sizes. This is in contrast to Holenstein's two-player result for no-signalling games,
which has no alphabet dependence at all \cite{Hol}. This is generalized
in our Theorem~\ref{th:repetitionSNOS}, since for two players we know by
Lemma \ref{lemma:2-player-SNOS} that $\NS \equiv \SNOS$,
and we could directly read off bounds with constants already improving on
Holenstein's. Looking a little into the proof allows us to optimize the
constants even more, which we record as follows.

\begin{theorem}[Parallel repetition of no-signalling $2$-player games, cf.~\cite{Hol}]
\label{thm:2-player}
Let $G$ be a $2$-player game with
$\omega_{\NS}(G)\leq 1-\delta$ for some $0<\delta<1$.
Then, for any $n\in\NN$, and any $t \geq (1-\delta+\alpha)n$ for some $0<\alpha\leq\delta$, we have
\begin{align*}
& \omega_{\NS}(G^n) \leq \left(1-\frac{\delta^2}{27}\right)^n,    \\
& \omega_{\NS}(G^{t/n}) \leq \exp\left( -n \frac{\alpha^2}{33} \right).
\end{align*}
\end{theorem}

\section{Constrained de Finetti reduction}
\label{sec:constrained-postselection}

De Finetti reductions are a useful tool when trying to understand any permutation-invariant information processing task. Indeed, these enable to restrict the analysis to that of i.i.d.~scenarios, which are usually trivially understood. In the context of multi-player games played $n$ times in parallel, one would like to use the fact that the numbering of the $n$ instances of the repeated game is irrelevant to reduce the study of strategies for the latter to the study of so-called
\emph{de Finetti strategies} (i.e.~convex combinations of $n$ i.i.d.~strategies).

The seminal de Finetti reduction (aka post-selection) lemma was stated in \cite{C-K-R}, later finding applications in many areas of quantum information theory, from quantum cryptography \cite{L-GP-R-C} to quantum Shannon theory \cite{B-C-R}. Our proofs though, will rely on two more recently established de Finetti reduction results, which are stated below. Just to fix some definitions: we will say that a (sub-)probability distribution $P_{\cZ^n}$, resp.~a conditional (sub-)probability distribution $P_{\cB^n|\cY^n}$, is $n$-symmetric if for any permutation $\pi$ of $n$ elements, $\forall\ z^n,\ P(\pi(z^n))=P(z^n)$, resp.~$\forall\ b^n,y^n,\ P(\pi(b^n)|\pi(y^n))=P(b^n|y^n)$.

\begin{lemma}[de Finetti reduction for conditional p.d.'s, \cite{AF-R}]
\label{lemma:dF-conditional}
Let $\cB,\cY$ be finite alphabets. There exists a probability measure $dR_{\cB|\cY}$ on the set
of conditional probability distributions $R_{\cB|\cY}$ such that, for any $n$-symmetric
conditional probability distribution $P_{\cB^n|\cY^n}$,
\[
  P_{\cB^n|\cY^n} \leq \mathrm{poly}(n) \int_{R_{\cB|\cY}} R_{\cB|\cY}^{\otimes n}\,\mathrm{d}R_{\cB|\cY},
\]
where the polynomial pre-factor may be upper bounded as $\mathrm{poly}(n)\leq (n+1)^{|\cB||\cY|}$.
\end{lemma}

\begin{lemma}[Constrained de Finetti reduction for (sub-)p.d.'s, \cite{L-W}]
\label{lemma:dF-constrained}
Let $\cZ$ be a finite alphabet. There exists a probability measure $dQ_{\cZ}$ on the
set of probability distributions $Q_{\cZ}$ on $\cZ$ such that, for any $n$-symmetric
(sub-)probability distribution $P_{\cZ^n}$ on $\cZ^n$,
\[
  P_{\cZ^n} \leq \mathrm{poly}(n) \int_{Q_{\cZ}} F\left(P_{\cZ^n},Q_{\cZ}^{\otimes n}\right)^2Q_{\cZ}^{\otimes n}\,\mathrm{d}Q_{\cZ},\]
where the polynomial pre-factor may be upper bounded as  $\mathrm{poly}(n)\leq (n+1)^{3|\cZ|^2}$.
\end{lemma}

In Lemma \ref{lemma:dF-constrained} above, as well as in the remainder of this paper, $F(P,Q)$ stands for the fidelity between probability distributions $P$ and $Q$, defined as $F(P,Q)=\|\sqrt{P}\sqrt{Q}\|_1$.

We are now ready to present the technical lemma that will allow us in Section \ref{sec:proofs} to reduce
the study of strategies for repeated games to the study of so-called
\emph{de Finetti strategies}, and hence prove our main results.

\begin{lemma}[de Finetti reduction for sub-no-signalling correlations]
\label{lemma:deFinetti}
There exists a probability measure $dQ$ on the set of probability distributions $Q$ on $\underline{\cA}\times\underline{\cX}$ such that for any probability distribution $T$ on $\underline{\cX}$ and
any $P\in \SNOS(\underline{\cA}^n|\underline{\cX}^n)$ an $n$-symmetric sub-no-signalling correlation,
it holds that
\begin{equation}
  \label{eq:definetti}
  T_{\underline{\cX}}^{\otimes n} P_{\underline{\cA}^n|\underline{\cX}^n}
       \leq \mathrm{poly}(n)
         \int_{Q_{\underline{\cA}\underline{\cX}}} \widetilde{F}\left(Q_{\underline{\cA}\underline{\cX}}\right)^{2n}
               Q_{\underline{\cA}\underline{\cX}}^{\otimes n}\,\mathrm{d}Q_{\underline{\cA}\underline{\cX}},
\end{equation}
where we defined
\[
  \widetilde{F}\left(Q_{\underline{\cA}\underline{\cX}}\right)
     := \underset{\emptyset \neq I\varsubsetneq[\ell]}{\min}\, \underset{R_{\cA_I|\cX_I}}{\max}\,
                 F\left(T_{\underline{\cX}}R_{\cA_I|\cX_I},Q_{\cA_I\underline{\cX}}\right).
\]
We mention for the sake of completeness that the $\mathrm{poly}(n)$
pre-factor in equation~\eqref{eq:definetti} may be upper bounded by
$(n+1)^{3|\underline{\cA}|^2|\underline{\cX}|^2+2|\underline{\cA}||\underline{\cX}|}$.
\end{lemma}

\begin{proof}
Since $T_{\underline{\cX}}^{\otimes n} P_{\underline{\cA}^n|\underline{\cX}^n}$ is an $n$-symmetric sub-probability distribution on $(\underline{\cA}\underline{\cX})^n$, we first of all have by Lemma~\ref{lemma:dF-constrained} that
\[ T_{\underline{\cX}}^{\otimes n} P_{\underline{\cA}^n|\underline{\cX}^n} \leq \mathrm{poly}(n) \int_{Q_{\underline{\cA}\underline{\cX}}} F\left(T_{\underline{\cX}}^{\otimes n}P_{\underline{\cA}^n|\underline{\cX}^n},Q_{\underline{\cA}\underline{\cX}}^{\otimes n}\right)^2Q_{\underline{\cA}\underline{\cX}}^{\otimes n}\,\mathrm{d}Q_{\underline{\cA}\underline{\cX}}. \]
Notice next that, for any $\emptyset \neq I\varsubsetneq[\ell]$,
\[
  F\left(T_{\underline{\cX}}^{\otimes n}P_{\underline{\cA}^n|\underline{\cX}^n},Q_{\underline{\cA}\underline{\cX}}^{\otimes n}\right) \leq F\left(T_{\underline{\cX}}^{\otimes n}P_{\cA_I^n|\underline{\cX}^n},Q_{\cA_I\underline{\cX}}^{\otimes n}\right) \leq F\left(T_{\underline{\cX}}^{\otimes n}P'_{\cA_I^n|\cX_I^n},Q_{\cA_I\underline{\cX}}^{\otimes n}\right).
\]
The first inequality is by monotonicity of the fidelity under stochastic
maps (in particular taking marginals). While the second inequality is
because $P\in \SNOS(\underline{\cA}^n|\underline{\cX}^n)$, so that
$P_{\cA_I^n|\underline{\cX}^n}\leq P'_{\cA_I^n|\cX_I^n}$ for some
conditional p.d.~$P'_{\cA_I^n|\cX_I^n}$, and because the fidelity is order-preserving.

What is more, for any $\emptyset \neq I\varsubsetneq[\ell]$, $P'_{\cA_I^n|\cX_I^n}$ can be chosen to be an $n$-symmetric conditional probability distribution. Indeed, if it were not, its $n$-symmetrization would still upper bound $P_{\cA_I^n|\underline{\cX}^n}$ (since the latter is by assumption $n$-symmetric). We then have by Lemma~\ref{lemma:dF-conditional} that
\[
  P'_{\cA_I^n|\cX_I^n} \leq \mathrm{poly}(n) \int_{R_{\cA_I|\cX_I}} R_{\cA_I|\cX_I}^{\otimes n}\,\mathrm{d}R_{\cA_I|\cX_I}, \]
and subsequently, using first, once more, that the fidelity is order-preserving, and second that it is multiplicative on tensor products,
\begin{align*}
F\left(T_{\underline{\cX}}^{\otimes n}P'_{\cA_I^n|\cX_I^n},Q_{\cA_I\underline{\cX}}^{\otimes n}\right) & \leq \mathrm{poly}(n) \underset{R_{\cA_I|\cX_I}}{\max}\, F\left(T_{\underline{\cX}}^{\otimes n} R_{\cA_I|\cX_I}^{\otimes n},Q_{\cA_I\underline{\cX}}^{\otimes n}\right)\\
& = \mathrm{poly}(n) \underset{R_{\cA_I|\cX_I}}{\max}\, F\left(T_{\underline{\cX}} R_{\cA_I|\cX_I},Q_{\cA_I\underline{\cX}}\right)^n.
\end{align*}

Recapitulating, we get
\[
  T_{\underline{\cX}}^{\otimes n} P_{\underline{\cA}^n|\underline{\cX}^n}
    \leq \mathrm{poly}(n) \int_{Q_{\underline{\cA}\underline{\cX}}} \left(\underset{\emptyset \neq I\varsubsetneq[\ell]}{\min}\, \underset{R_{\cA_I|\cX_I}}{\max}\, F\left(T_{\underline{\cX}}R_{\cA_I|\cX_I},Q_{\cA_I\underline{\cX}}\right)\right)^{2n}Q_{\underline{\cA}\underline{\cX}}^{\otimes n}\,\mathrm{d}Q_{\underline{\cA}\underline{\cX}},
\]
as announced.
\end{proof}

\section{Proofs of the main Theorems}
\label{sec:proofs}

In this section we prove Theorem \ref{th:repetitionSNOS}, Corollary~\ref{cor:repetitionSNOS} and Theorem \ref{thm:2-player}.

We need first of all the following extension of Lemma 9.5 in \cite{Hol}:
\begin{lemma}
  \label{lemma:Hol-Lemma-9.5}
  For $\underline{\cZ} = \bigtimes_{\!\!j=1}^m\, \cZ_j$ and
  $\underline{\cB} = \bigtimes_{\!\!j=1}^m\, \cB_j$,
  consider probability distributions $T$ on $\underline{\cZ}$
  and $P$ on $\underline{\cB}\times\underline{\cZ}$ satisfying
  \begin{equation} \label{eq:markov1}
    \frac12 \bigl\| P_{\underline{\cZ}} - T_{\underline{\cZ}} \bigr\|_1 \leq \epsilon_0.
  \end{equation}
  If for each $j\in[m]$ there exists a conditional probability distribution
  $Q(b_j|z_j)$ such that
  \begin{equation} \label{eq:markov2}
   \frac12 \bigl\| P_{\cB_j\underline{\cZ}} - T_{\underline{\cZ}}Q_{\cB_j|\cZ_j} \bigr\|_1 \leq \epsilon_j,
  \end{equation}
  then there exists a conditional probability distribution
  $P'(\underline{b}|\underline{z})$ such that, for each $j\in[m]$, $P'(b_j|\underline{z}) = P'(b_j|z_j)$ for all $b_j,\underline{z}$, and
  \begin{equation} \label{eq:reconstruct}
  \frac12 \bigl\| T_{\underline{\cZ}} P'_{\underline{\cB}|\underline{\cZ}} - P_{\underline{\cB}\underline{\cZ}} \bigr\|_1 \leq \epsilon_0 + \sum_{j=1}^m 2\epsilon_j.
  \end{equation}
\end{lemma}

\begin{proof}
This works exactly as the proofs of the case $m=2$, appearing as Lemma 9.5 in \cite{Hol}, or of the case $m=3$, appearing as Lemma 5.4 in \cite{Ros}. Both statements follow from applying either two or three times Lemma 9.4 in \cite{Hol}. Let us state the latter for completeness, and then only sketch how the proofs of the cases $m=2$ or $m=3$ generalize to any $m$.
\begin{quote} Holenstein (\cite{Hol}, Lemma 9.4): Let $\overline{P}_{\mathcal{S}\mathcal{T}}$ and $\overline{Q}_{\mathcal{S}}$ be probability distributions over $\mathcal{S}\times\mathcal{T}$ and $\mathcal{S}$ respectively. There exists a probability distribution $\overline{R}_{\mathcal{S}\mathcal{T}}$ over $\mathcal{S}\times\mathcal{T}$ such that
\[ \bigl\| \overline{R}_{\mathcal{S}\mathcal{T}} - \overline{P}_{\mathcal{S}\mathcal{T}} \bigr\|_1 \leq  \bigl\| \overline{Q}_{\mathcal{S}} - \overline{P}_{\mathcal{S}} \bigr\|_1\ \text{and}\  \bigl\| \overline{R}_{\mathcal{S}} - \overline{Q}_{\mathcal{S}} \bigr\|_1= \bigl\| \overline{R}_{\mathcal{T}} - \overline{P}_{\mathcal{T}} \bigr\|_1=0. \]
\end{quote}
Thanks to this result, we know that we can recursively construct a sequence $P^{(1)},\ldots,P^{(m)}$ of probability distributions on $\underline{\mathcal{B}}\times\underline{\mathcal{Z}}$ such that, setting $P^{(0)}=Q$, for each $j\in[m]$, we have: for any fixed $\underline{z}\in\underline{\mathcal{Z}}$,
\begin{align*}
& \bigl\| P^{(j)}_{\underline{\mathcal{B}}|\underline{\mathcal{Z}}=\underline{z}} - P^{(j-1)}_{\underline{\mathcal{B}}|\underline{\mathcal{Z}}=\underline{z}} \bigr\|_1 \leq \bigl\| P_{\mathcal{B}_j|\underline{\mathcal{Z}}=\underline{z}} - Q_{\mathcal{B}_j|\mathcal{Z}_j=z_j} \bigr\|_1, \\
& \bigl\| P^{(j)}_{\mathcal{B}_j|\underline{\mathcal{Z}}=\underline{z}} - Q_{\mathcal{B}_j|\mathcal{Z}_j=z_j} \bigr\|_1 = 0, \\
& \forall\ k\in[m]\setminus\{j\},\ \bigl\| P^{(j)}_{\mathcal{B}_k|\underline{\mathcal{Z}}=\underline{z}} - Q_{\mathcal{B}_k|\underline{\mathcal{Z}}=\underline{z}} \bigr\|_1 = 0.
\end{align*}
The probability distribution $P^{(m)}$ then satisfies
\[ \bigl\|T_{\underline{\cZ}} P^{(m)}_{\underline{\cB}|\underline{\cZ}} - P_{\underline{\cB}\underline{\cZ}} \bigr\|_1 \leq \bigl\| P_{\underline{\cZ}} - T_{\underline{\cZ}} \bigr\|_1 + \sum_{j=1}^m 2\bigl\| P_{\cB_j\underline{\cZ}} - T_{\underline{\cZ}}Q_{\cB_j|\cZ_j} \bigr\|_1, \]
and can therefore be chosen as the desired $P'$.
\end{proof}

We just mention as a side note that Lemma 9.4 in \cite{Hol} crucially relies on the following fact: the statistical distance between two probability distributions $P_1,P_2$, i.e.
\[ \frac{1}{2}\|P_1-P_2\|_1, \]
can be equivalently characterized as the minimum probability that $X_1$ differs from $X_2$ over pairs of random variables $(X_1,X_2)$ sampled from $P$ having $(P_1,P_2)$ as marginals.

Note that the conditions enforced in Lemma \ref{lemma:Hol-Lemma-9.5}
are not enough to ensure no-signalling of $P'$ for three or more players. They would be sufficient though to guarantee that $P'$ satisfies the relaxed no-signalling constraints considered in \cite{Ros}, namely that any group of $\ell-1$ players together cannot signal to the remaining player. In other words, if a correlation approximately satisfies the Markov chain conditions necessary for being no-signalling, in the form of equations \eqref{eq:markov1} and \eqref{eq:markov2}, then it is approximated, in the sense of equation \eqref{eq:reconstruct}, by a ``weak'' no-signalling correlation, as considered in \cite{Ros}. Nevertheless, we can leverage this result to approximate the given
no-signalling correlation by a sub-no-signalling correlation.

\begin{lemma}
  \label{lemma:Hol-SNOS}
  Let $P$ be a probability distribution on $\underline{\cA}\times\underline{\cX}$
  and $T$ be a probability distribution on $\underline{\cX}$. If the no-signalling conditions
  (\ref{eq:NS}) hold approximately, namely
  \[
    \forall\ I\subsetneq[\ell],\ \exists\ Q(\cdot|x_I)\text{ p.d.'s on }\cA_I\text{ s.t.~}
      \frac12 \bigl\| P_{\cA_I\underline{\cX}}
                      - T_{\underline{\cX}} Q_{\cA_I|\cX_I} \bigr\|_1 \leq \epsilon_I,
  \]
  then there exists a sub-no-signalling correlation $P' \in \SNOS(\underline{\cA}|\underline{\cX})$
  that approximates $P$, in the sense that
  \[
    \frac12 \bigl\| T_{\underline{\cX}}P'_{\underline{\cA}|\underline{\cX}}
                                         - P_{\underline{\cA}\underline{\cX}} \bigr\|_1
                 \leq \epsilon_{\emptyset} + \sum_{\emptyset\neq I\subsetneq[\ell]} 2\epsilon_I.
  \]
  In the two-player case $\ell=2$, $P'$ can be chosen to be
  no-signalling itself, $P'\in \NS(\underline{\cA}|\underline{\cX})$.
\end{lemma}

\begin{proof}
We will apply Lemma~\ref{lemma:Hol-Lemma-9.5}, with $m=2^\ell-2$,
the index $j$ identifying a non-empty and non-full set
$\emptyset \neq I \subsetneq [\ell]$ (for instance via the
expansion of $j$ into $\ell$ binary digits). The local input and
output alphabets are
\[
  \cZ_j = \bigtimes_{i\in I} \cX_i,\quad \cB_j = \bigtimes_{i\in I} \cA_i,
\]
and the distribution we apply it to is
\[
  \widehat{P}(\underline{b}\underline{z})
   = \begin{cases}
       P(\underline{a}\underline{x}) & \text{ if } \forall j,\ b_j = (a_i:i\in I),\ z_j = (x_i:i\in I), \\
       0                             & \text{ otherwise}.
     \end{cases}
\]
Likewise, the prior distribution on $\underline{\cZ}$ is given by
\[
  \widehat{T}(\underline{z})
   = \begin{cases}
       T(\underline{x}) & \text{ if } \forall j,\ z_j = (x_i:i\in I), \\
       0                & \text{ otherwise},
     \end{cases}
\]
and we use the conditional distributions $Q(b_j|z_j) = Q(a_I|x_I)$.

Now, the prerequisites of Lemma~\ref{lemma:Hol-Lemma-9.5} are given,
with $\epsilon_j=\epsilon_I$, and thus we get a conditional probability
distribution $\widehat{P}'$ with $\widehat{P}'(b_j|\underline{z}) = \widehat{P}'(b_j|z_j)$
for all $j$, and
\[
  \frac12 \bigl\| \widehat{T}_{\underline{\cZ}}\widehat{P}'_{\underline{\cB}|\underline{\cZ}}
                                       - \widehat{P}_{\underline{\cB}\underline{\cZ}} \bigr\|_1
                                               \leq \epsilon_0 + \sum_{j=1}^n 2\epsilon_j =: \epsilon.
\]
We would like to conclude here by ``pulling back'' this conditional
distribution to a correlation on
$\underline{\cA}\times\underline{\cX}$, which we would wish to be no-signalling. This almost works, except that $P'$ has support
outside the image of the diagonal embedding
\begin{align*}
  \Delta : \underline{\cA} &\longrightarrow \underline{\cB} \\
             \underline{a} &\longmapsto     \underline{b} \text{ s.t.~} \forall j,\ b_j = (a_i:i\in I),
\end{align*}
and likewise for $\Delta:\underline{\cX} \longrightarrow \underline{\cZ}$.

To resolve this issue, we simply remove this part of the distribution,
and define the desired sub-normalized conditional densities by letting
\[
  P'(\underline{a}|\underline{x})
       := \widehat{P}'\bigl(\Delta(\underline{a})|\Delta(\underline{x})\bigr).
\]
From this we see directly that
\[
  \frac12 \bigl\| T_{\underline{\cX}} P'_{\underline{\cA}|\underline{\cX}}
                                       - P_{\underline{\cA}\underline{\cX}} \bigr\|_1
                                                                            \leq \epsilon,
\]
because $\widehat{P}(\underline{b},\underline{z}) = P(\underline{a},\underline{x})$
for $\underline{b} = \Delta(\underline{a})$ and $\underline{z} = \Delta(\underline{x})$,
and it is $0$ outside the image of $\Delta$.

It remains to check that $P'$ is sub-no-signalling. Let
$\emptyset \neq I \subsetneq [\ell]$ be a subset with corresponding index
$1\leq j \leq 2^\ell-2$. Let also $\underline{x}\in\underline{\cX}$,
$a_I\in \cA_I$ be tuples, and set $\underline{z} = \Delta(\underline{x})$,
$\underline{b} = \Delta(\underline{a})$ (so that $z_j=x_I\in \cX_I = \cZ_j$, $b_j = a_I \in \cA_I = \cB_j$).
Then,
\begin{align*}
P'(a_I|\underline{x}) = & \sum_{a_{I^c}\in\cA_{I^c}} P'(\underline{a}|\underline{x}) \\
= & \sum_{a_{I^c}\in\cA_{I^c}} \widehat{P}'\bigl(\Delta(\underline{a})|\Delta(\underline{x})\bigr) \\
\leq & \sum_{b_k \in \cB_k,\ k\neq j} \widehat{P}'\bigl(\underline{b}|\underline{z}\bigr)\\
= & \,\widehat{P}'(b_j|\underline{z}\bigr) \\
= & \,\widehat{P}'(b_j|z_j) =:   Q'(a_I|x_I).
\end{align*}
Here, we have used the definition of the marginal and of $P'$.
The inequality in the third line is because we enlarge the domain
of the summation, and the equality in the last line is by the marginal property
of $\widehat{P}'$.

The last claim, regarding $\ell=2$ players, is the original
Lemma 9.5 in \cite{Hol}.
\end{proof}

We are now ready to prove our main theorem, namely the parallel
repetition and concentration results for the sub-no-signalling value of multi-player games.

\begin{proof}[Proof of Theorem~\ref{th:repetitionSNOS}]
Let $P_{\underline{\cA}^n|\underline{\cX}^n}$ be a sub-no-signalling correlation which is optimal to win the game $G^n$. The distribution $T_{\underline{\cX}}^{\otimes n}$ and the predicate $V_{\underline{\cA}\underline{\cX}}^{\otimes n}$ of $G^n$ being $n$-symmetric, we can assume without loss of generality that $P_{\underline{\cA}^n|\underline{\cX}^n}$ is also $n$-symmetric. Indeed, since for any permutation $\pi$ of $n$ elements, $T\circ\pi=T$ and $V\circ\pi=V$, playing $G^n$ with $P$ or with $P\circ\pi$ yields the same winning probability. And therefore, if $P$ is an optimal strategy then so is its symmetrization over all permutations of $n$ elements. Hence, by Lemma \ref{lemma:deFinetti},
\[
  T_{\underline{\cX}}^{\otimes n} P_{\underline{\cA}^n|\underline{\cX}^n}
      \leq \mathrm{poly}(n) \int_{Q_{\underline{\cA}\underline{\cX}}} \widetilde{F}\left(Q_{\underline{\cA}\underline{\cX}}\right)^{2n}Q_{\underline{\cA}\underline{\cX}}^{\otimes n}\,\mathrm{d}Q_{\underline{\cA}\underline{\cX}}.
\]
Now, fix $0<\epsilon<1$ and define
\[
  \mathcal{P}_{\epsilon}
    :=\left\{ Q_{\underline{\cA}\underline{\cX}}\ :\ \max_{\emptyset \neq I\varsubsetneq[\ell]} \min_{R_{\cA_I|\cX_I}}\, \frac{1}{2}\|T_{\underline{\cX}}R_{\cA_I|\cX_I}-Q_{\cA_I\underline{\cX}}\|_1 \leq \epsilon \right\}.
\]
Observe that, by well-known relations between fidelity and trace-distance (see e.g.~\cite{F-vdG}), if $Q_{\underline{\cA}\underline{\cX}}\notin\mathcal{P}_{\epsilon}$,
then automatically $\widetilde{F}\left(Q_{\underline{\cA}\underline{\cX}}\right)^2\leq 1-\epsilon^2$.
Hence,
\[
  T_{\underline{\cX}}^{\otimes n} P_{\underline{\cA}^n|\underline{\cX}^n}
     \leq \mathrm{poly}(n)\left( \int_{Q_{\underline{\cA}\underline{\cX}}\in\mathcal{P}_{\epsilon}} Q_{\underline{\cA}\underline{\cX}}^{\otimes n}\,\mathrm{d}Q_{\underline{\cA}\underline{\cX}} + (1-\epsilon^2)^{n}\int_{Q_{\underline{\cA}\underline{\cX}}\notin\mathcal{P}_{\epsilon}} Q_{\underline{\cA}\underline{\cX}}^{\otimes n}\,\mathrm{d}Q_{\underline{\cA}\underline{\cX}} \right).
\]
On the other hand, if $Q_{\underline{\cA}\underline{\cX}}\in\mathcal{P}_{\epsilon}$,
then by definition
\[
  \forall\ \emptyset\neq I\varsubsetneq[\ell],\ \exists\ R_{\cA_I|\cX_I}:\
  \frac{1}{2}\|T_{\underline{\cX}}R_{\cA_I|\cX_I}-Q_{\cA_I\underline{\cX}}\|_1 \leq \epsilon.
\]
By Lemma \ref{lemma:Hol-SNOS}, the latter condition implies that there exists
a sub-no-signalling correlation $R'_{\underline{\cA}|\underline{\cX}}$ such that
\[ \frac{1}{2}\|T_{\underline{\cX}}R'_{\underline{\cA}|\underline{\cX}}-Q_{\underline{\cA}\underline{\cX}}\|_1 \leq C_{\ell}\epsilon,\ \text{where}\ C_{\ell}=1+2(2^{\ell}-2)=2^{\ell+1}-3. \]
Yet, the winning probability when playing $G$ with a strategy $R'_{\underline{\cA}|\underline{\cX}}\in \SNOS(\underline{\cA}|\underline{\cX})$ is, by assumption on $G$, at most $1-\delta$. So the average of the predicate of $G$ over $Q_{\underline{\cA}\underline{\cX}}\in\mathcal{P}_{\epsilon}$ is at most $1-\delta+2C_{\ell}\epsilon$.
Putting everything together, we eventually get that the winning probability when playing $G^n$ with strategy $P_{\underline{\cA}^n|\underline{\cX}^n}$ is upper bounded as
\begin{equation}
  \label{eq:split}
  \PP(\text{win}) \leq \mathrm{poly}(n)\left( (1-\delta+2C_{\ell}\epsilon)^n + (1-\epsilon^2)^n \right).
\end{equation}
Choosing in equation~\eqref{eq:split}
\[ \epsilon=C_{\ell}\left(\left(1+\frac{\delta}{C_{\ell}^2}\right)^{1/2}-1\right)\geq \frac{99\delta}{200C_{\ell}},\ \text{so that}\ \epsilon^2\geq\frac{\delta^2}{5C_{\ell}^2}, \]
and recalling that $P_{\underline{\cA}^n|\underline{\cX}^n}$ is, by hypothesis, an optimal sub-no-signalling strategy, we obtain
\begin{equation}
  \label{eq:poly-exp}
  \omega_{\SNOS}(G^n) \leq \mathrm{poly}(n) \left(1-\frac{\delta^2}{5C_{\ell}^2}\right)^n.
\end{equation}

In order to conclude, we have to remove the polynomial pre-factor. So assume that there exists a constant
$C>0$ such that for some $N\in\NN$, $\omega_{\SNOS}(G^N)\geq C\left(1-\delta^2/5C_{\ell}^2\right)^N$.
Then, for any $n\in\NN$, we would have
\[ \omega_{\SNOS}(G^{Nn}) \geq \left(\omega_{\SNOS}(G^{N})\right)^{n}
                          \geq C^{n}\left(1-\frac{\delta^2}{5C_{\ell}^2}\right)^{Nn}. \]
On the other hand, however, we still have by equation~\eqref{eq:poly-exp}
\[ \omega_{\SNOS}(G^{Nn}) \leq \mathrm{poly}(Nn) \left(1-\frac{\delta^2}{5C_{\ell}^2}\right)^{Nn}. \]
Letting $n$ grow, we see that the only option to make these two conditions compatible is to have $C\leq 1$, which is precisely what we wanted to show.

Following the exact same lines as above, we also get the
concentration bound.
Indeed, for any $t\geq (1-\delta+\alpha)n$, we now have in place of equation~\eqref{eq:split} that, for any $0<\epsilon<1$,
\begin{equation} \label{eq:split-t}
  \omega_{\SNOS}(G^{t/n})
    \leq \mathrm{poly}(n)\left( \exp\left[-2n(\alpha-2C_{\ell}\epsilon)^2\right]
                                          + \exp\left[-n\epsilon^2\right] \right).
\end{equation}
The first term in the r.h.s.~of equation \eqref{eq:split-t} is a consequence of Hoeffding's inequality: if $\underline{A},\underline{X}$ are distributed according to $Q_{\underline{\cA}\underline{\cX}}\in\mathcal{P}_{\epsilon}$, then the value of the game predicate is on average at most $1-\delta+2C_{\ell}\epsilon$, so for $n$ independent such $\underline{A},\underline{X}$, the probability that the sum of the $n$ values of the game predicate is above $(1-\delta+\alpha)n$ is at most $\exp[-2n(\alpha-2C_{\ell}\epsilon)^2]$. The second term in the r.h.s.~of equation \eqref{eq:split-t} is obtained by simply using that $e^{-x}\geq 1-x$ for any $x>0$.

The announced upper bound follows from choosing in equation~\eqref{eq:split-t}
\[ \epsilon = \frac{(4C_{\ell}-\sqrt{2})\alpha}{8C_{\ell}^2-1} \geq \frac{5(20-\sqrt{2})\alpha}{199C_{\ell}},\ \text{so that}\ \epsilon^2\geq\frac{\alpha^2}{5C_{\ell}^2}, \]
and removing the polynomial pre-factor by the same trick as before.
\end{proof}

\begin{proof}[Proof of Corollary~\ref{cor:repetitionSNOS}]
By Lemma \ref{lemma:SNOS-NS-multiplayers}, we know that if $G$ is
an $\ell$-player game with full support satisfying
$\omega_{\NS}(G)\leq 1-\delta$,
then $\omega_{\SNOS}(G)\leq 1-\delta/(\Gamma+1)$.
And thus by Theorem \ref{th:repetitionSNOS},
\[
  \omega_{\NS}(G^n) \leq \omega_{\SNOS}(G^n) \leq \left(1-\frac{\delta^2}{5C_{\ell}^2(\Gamma+1)^2}\right)^n.
\]
The concentration bound for $\omega_{\NS}(G^{t/n})$ follows analogously.
\end{proof}

\begin{proof}[Proof of Theorem~\ref{thm:2-player}]
We follow the exact same reasoning as in the proof of Theorem \ref{th:repetitionSNOS}, and keep the same notation. In the case $\ell=2$, we have by Lemma \ref{lemma:Hol-SNOS} that, for any $0<\epsilon<1$,
\[ Q_{\underline{\cA}\underline{\cX}}\in\mathcal{P}_{\epsilon}\ \Rightarrow\ \exists\ R'_{\underline{\cA}\underline{\cX}}\in\NS(\underline{\cA}|\underline{\cX}):\  \frac{1}{2}\|T_{\underline{\cX}}R'_{\underline{\cA}|\underline{\cX}}-Q_{\underline{\cA}\underline{\cX}}\|_1 \leq 5\epsilon. \]
Yet, if the winning probability when playing $G$ with a strategy $R'_{\underline{\cA}|\underline{\cX}}\in \NS(\underline{\cA}|\underline{\cX})$ is, by assumption on $G$, at most $1-\delta$, then the average of the predicate of $G$ over $Q_{\underline{\cA}\underline{\cX}}\in\mathcal{P}_{\epsilon}$ is at most $1-\delta+5\epsilon$. This is because we are here dealing with normalised probability distributions. Hence, for any $0<\epsilon<1$,
\begin{align*}
& \omega_{\NS}(G^n) \leq \mathrm{poly}(n)\left( (1-\delta+5\epsilon)^n + (1-\epsilon^2)^n \right),\\
& \omega_{\SNOS}(G^{t/n}) \leq \mathrm{poly}(n)\left( \exp\left[-2n(\alpha-5\epsilon)^2\right] + \exp\left[-n\epsilon^2\right] \right).
\end{align*}
We can now choose $\epsilon=(\sqrt{29}-5)\delta/2$ in the parallel repetition
estimate and $\epsilon=(10-\sqrt{2})\alpha/49$ in the concentration bound one,
and argue as in the proof of Theorem \ref{th:repetitionSNOS} to remove the
polynomial pre-factor, which yields the two advertised results.
\end{proof}

\section{Discussion}
\label{sec:end}

Our main contribution in the present paper is a concentration result for the
sub-no-signalling value of multi-player games under parallel repetition. In
fact, we believe that our work is the first to recognize the intrinsic interest
of the class of sub-no-signalling correlations, which appears naturally as a
relaxation of the no-signalling ones. In particular, the fact that sub-no-signalling correlations have total probability
less than or equal to $1$ can be interpreted as the possibility
of ``abstaining'' from giving an answer in $\underline{\cA}$, with
a certain probability depending on the input in $\underline{\cX}$.
However, each marginal $P_{\cA_I|\underline{\cX}}$ has to be
consistent ``locally'' with the no-signalling behaviour, in that
it has to be dominated by a correlation $Q_{\cA_I|\cX_I}$ that
depends only on the $I$ positions of the input $\underline{x}$.
This means that each group $I$ of players is able to interpret their
observed statistics as being ``really'' governed by a local
marginal $Q_{\cA_I|\cX_I}$, only that sometimes the device
generating the correlation defaults and does not give an answer.
The probability of abstention depends on the entire input
$\underline{x}$, and thus would be signalling, if observed.
Indeed, the anti-correlation game discussed in Example \ref{ex} shows that a sub-no-signalling correlation cannot always be embedded
in a no-signalling one, except in the case of two players
(cf. Lemma \ref{lemma:2-player-SNOS}). Abstention thus gives more power in general, but it
comes with a price as well, since abstaining does not mean winning
the game. In this sense, it should not be confused with plain
post-selection (that is, conditioning) on the non-abstaining
event, which is well-known to allow the violation of Bell
inequalities by otherwise local hidden variables, via the
so-called ``detection loophole'' \cite{G-M,Ebe}.

Specifically, if an $\ell$-player game $G$ has $\SNOS$ value $1-\delta$, then the
probability for $\SNOS$ players to win a fraction at least $1-\delta+\alpha$ of
$n$ instances of $G$ played in parallel is at most $\exp(-nC_{\ell}\alpha^2)$,
where $C_{\ell}>0$ is a constant which only depends on the number $\ell$ of
players. This, a universal multi-player parallel repetition and concentration
bound, is in contrast to the results on \cite{B-F-S} and \cite{AF-R-V},
which are restricted to full-support game distributions and
with constants that seem to depend heavily on the game. We
think of these findings as evidence that sub-no-signalling
correlations are natural, due to their well-behaved parallel
repetition properties. As hinted at in \cite{B-F-S}, such a result,
valid for games involving strictly more than $2$ players and where not all queries are asked \cite{Harry}, might potentially find applications in position-based cryptography \cite{B-C-F-G-G-O-S,F-K-T-W}. It would also be interesting
to investigate whether the recent work of \cite{K-R-R}, showing multi-prover interactive proofs for EXP (exponential
time languages) that are robust against no-signalling provers, remains valid for sub-no-signalling provers, and whether our
result can be generalized to amplify the soundness gap of their
scheme. The latter is not self-evident, as they require a polynomial
number of provers, but our bounds carry a penalty exponential
in the number of players.

In the case $\ell=2$, our concentration statement is actually equivalent to the analogous one for the no-signalling value of $G$, thus with a universal constant $c=C_2$ in the exponential bound. And we know we cannot hope for a
better dependence in $\alpha$ than the obtained one, even in the special case
$\alpha=\delta$, as proved in \cite{K-R}.
In the case $\ell>2$, our result implies a concentration
bound for the no-signalling value of $G$, but only if its input distribution has
full support. Besides, the constant in the exponential bound is this time highly
game-dependent (dependence on the sizes of the input and output alphabets, and
on the smallest weight occurring in the input distribution). This is fully
comparable to previous work in this direction due to Buhrman, Fehr and
Schaffner~\cite{B-F-S}, and Arnon-Friedman, Renner and Vidick~\cite{AF-R-V}.

Hence, the most immediate open problem at that point is regarding games with
non-full support in the case of three or more players (e.g.~the anti-correlation game): does a parallel repetition
result hold for the no-signalling value of such multi-player games?
Answering this question probably requires to understand first whether in
Corollary~\ref{cor:repetitionSNOS}, the presence of the game parameter $\Gamma$
is really necessary or is just an artifact of the proof technique. In other words,
does the rate at which the no-signalling value of a game decays under parallel
repetition truly depends on the game distribution?

Another issue that would be worth investigating is whether constrained de Finetti
reductions could also be used to establish parallel repetition results for the
classical or quantum value of multi-player games. Formally, the sets of classical
correlations $\text{C}(\underline{\cA}|\underline{\cX})$ and quantum correlations
$\text{Q}(\underline{\cA}|\underline{\cX})$ are defined as follows:
\[
  P \in \text{C}(\underline{\cA}|\underline{\cX})
      :\Leftrightarrow
      \forall\ \underline{x},\underline{a},\
        P(\underline{a}|\underline{x})
           =\sum_{m\in\mathcal{M}}Q(m)P_1(a_1|x_1\,m)\cdots P_{\ell}(a_{\ell}|x_{\ell}\,m),
\]
for some p.d.~$Q$ on some alphabet $\mathcal{M}$ and some p.d.'s $P_i(\cdot|x_i\,m)$ on $\cA_i$.
\[
  P \in \text{Q}(\underline{\cA}|\underline{\cX})
      :\Leftrightarrow
      \forall\ \underline{x},\underline{a},\
         P(\underline{a}|\underline{x})
            =\bra{\psi}M(x_1)_{a_1}\otimes\cdots\otimes M(x_{\ell})_{a_{\ell}}\ket{\psi},
\]
for some pure state $\ketbra{\psi}{\psi}$ on $\mathrm{H}_1\otimes\cdots\otimes\mathrm{H}_{\ell}$ and some POVMs $M(x_i)$ on $\mathrm{H}_i$.\\
And the classical, resp.~quantum, value of an $\ell$-player game $G$ with distribution $T$ and predicate $V$, denoted $\omega_{\text{C}}(G)$, resp.~$\omega_{\text{Q}}(G)$, is then naturally defined as the maximum, resp.~supremum, of the winning probability
\[ \label{eq:NS-SNOS-value}
  \PP\left(\text{win}\right) = \sum_{\underline{a},\underline{x}} T(\underline{x}) V(\underline{a},\underline{x}) P(\underline{a}|\underline{x}) \]
over all $P\in\text{C}(\underline{\cA}|\underline{\cX})$, resp.~$P\in\text{Q}(\underline{\cA}|\underline{\cX})$.

In the classical case, the first parallel repetition result for two-player games
was established by Raz~\cite{Raz}, and later improved by Holenstein~\cite{Hol},
while Rao~\cite{Rao} gave a concentration bound. However, the proof techniques are
arguably not as straightforward as via de Finetti reductions, and do not
generalise directly to any number $\ell$ of players. In the quantum case, even
less is known. The best parallel repetition result up to now is the one established by
Chailloux and Scarpa~\cite{C-S} (subsequently improved by Chung, Wu and Yuen~\cite{C-W-Y}),
which applies to two-player ($\ell$-player) free games, and from there to games
with full support.
That is why being able to export
ideas from the de Finetti approach to these two cases would be of great interest.
Roughly speaking, the problem we are facing is the following: Given an $n$-symmetric
correlation $P_{\underline{\cA}^n|\underline{\cX}^n}$, we can always write the
first step in the proof of Lemma~\ref{lemma:deFinetti}, i.e.~\begin{equation} \label{eq:dFstrategy}
  T_{\underline{\cX}}^{\otimes n} P_{\underline{\cA}^n|\underline{\cX}^n}
       \leq \mathrm{poly}(n)
            \int_{Q_{\underline{\cA}\underline{\cX}}}
                    F\left(T_{\underline{\cX}}^{\otimes n}P_{\underline{\cA}^n|\underline{\cX}^n},
                           Q_{\underline{\cA}\underline{\cX}}^{\otimes n}\right)^2
                    Q_{\underline{\cA}\underline{\cX}}^{\otimes n}\,\mathrm{d}Q_{\underline{\cA}\underline{\cX}}.
\end{equation}
Now, we would like to argue that if $P_{\underline{\cA}^n|\underline{\cX}^n}$ is a
classical, resp.~quantum, correlation, then the p.d.'s $Q_{\underline{\cA}\underline{\cX}}$
for which the fidelity weight in the r.h.s.~of equation~\eqref{eq:dFstrategy} is not exponentially small are
necessarily close to being of the form $T_{\underline{\cX}}R_{\underline{\cA}|\underline{\cX}}$
for some classical, resp.~quantum, correlation $R_{\underline{\cA}|\underline{\cX}}$.
This was precisely our proof philosophy in the no-signalling case. However, the
fact that the classical and quantum conditions are not properties that one can
read off on the marginals, contrary to the no-signalling one, seems to be a first
obstacle to surmount.

One related legitimate question would be the following: is it possible to make an even stronger statement than the one that, as explained above, we either are looking for (in the classical and quantum cases) or already have (in the no-signalling case)? Namely, could we upper bound $T_{\underline{\cX}}^{\otimes n}P_{\underline{\cA}^n|\underline{\cX}^n}$ by a de Finetti distribution analogous to that in the r.h.s.~of equation~\eqref{eq:dFstrategy}, but with weight strictly $0$ on p.d.'s $Q_{\underline{\cA}\underline{\cX}}$
which are not of the form $T_{\underline{\cX}}R_{\underline{\cA}|\underline{\cX}}$, for $R_{\underline{\cA}|\underline{\cX}}$ belonging to the same class as $P_{\underline{\cA}^n|\underline{\cX}^n}$? The answer to this question is no. Indeed, such improved de Finetti reduction would imply a strong parallel repetition result, which we know does not hold (see \cite{AF-R-V} for a similar discussion). So the best we can hope for is really to show that the fidelity weight in our upper bounding de Finetti distribution is exponentially small on the  p.d.'s which are too far from being of the desired form.

Finally, let us briefly comment on the main spirit difference between the present work and the one by Arnon-Friedman \emph{et al.}~\cite{AF-R-V}. Our approach consists in using a more ``flexible'' de Finetti reduction, in which the information on the correlation $P_{\underline{\cA}^n|\underline{\cX}^n}$ and the p.d.~$T_{\underline{\cX}}^{\otimes n}$ of interest are kept in the upper bounding de Finetti distribution, through the fidelity weight $F(T_{\underline{\cX}}^{\otimes n}P_{\underline{\cA}^n|\underline{\cX}^n},Q_{\underline{\cA}\underline{\cX}}^{\otimes n})^2$. Whereas in \cite{AF-R-V}, any initial correlation is first upper bounded by the same universal de Finetti correlation, on which a test (specifically tailored to the considered game distribution) is performed in a second step, that has the property of letting pass, resp.~rejecting, with high probability the strategies which are no-signalling, resp.~too signalling. So it seems in the end that both approaches are quite closely related: in our case, the ``signalling test'' which is applied to a given p.d.~$Q_{\underline{\cA}\underline{\cX}}$ is nothing else than the maximal fidelity of $Q_{\underline{\cA}\underline{\cX}}$ to the set of p.d.'s of the form $T_{\underline{\cX}}R_{\underline{\cA}|\underline{\cX}}$, with $R_{\underline{\cA}|\underline{\cX}}$ no-signalling, being above or below a certain threshold value. Also, it would be interesting (and potentially fruitful) to investigate whether one could combine in some way the techniques yielding Lemmas \ref{lemma:dF-conditional} and \ref{lemma:dF-constrained}, to get a de Finetti reduction result that would have the advantages of both: namely, that is designed for conditional p.d.'s while at the same carrying the relevant information on the conditional p.d.~it is applied to.

\section*{Acknowledgements}
We thank Rotem Arnon-Friedman, Harry Buhrman and Renato Renner for
insightful discussions on parallel repetition and de Finetti theorems.

This research was supported by the European Research Council
(Advanced Grant IRQUAT, ERC-2010-AdG-267386), the European Commission
(STREP RAQUEL, FP7-ICT-2013-C-323970), the Spanish MINECO
(projects FIS2008-01236 and FIS2013-40627-P), with the support of FEDER
funds, the Generalitat de Catalunya (CIRIT project 2014-SGR-966),
and the French CNRS (ANR projects OSQPI 11-BS01-0008 and Stoq 14-CE25-0033).


\begin{thebibliography}{8}
  \bibitem{AF-R} R.~Arnon-Friedman, R.~Renner,
    ``de Finetti reductions for correlations'',
    \emph{J. Math. Phys.} {\bf 56}:052203 (2015);
    arXiv[quant-ph]:1308.0312.

  \bibitem{AF-R-V} R.~Arnon-Friedman, R.~Renner, T.~Vidick,
    ``Non-signalling parallel repetition using de Finetti reductions'', \emph{IEEE Transactions on Information Theory} {\bf 62}(3) (2016);
    arXiv[quant-ph]:1411.1582.

  \bibitem{B-L-M-P-P-R} J.~Barrett, N.~Linden, S.~Massar, S.~Pironio, S.~Popescu, D.~Roberts, ``Non-local correlations as an information theoretic resource'', \emph{Physical Review A} {\bf 71}(2):022101 (2005); arXiv:quant-ph/0404097.

  \bibitem{Bell} J.~Bell, ``On the Einstein Podolsky Rosen paradox'', \emph{Physics} {\bf 1}:195 (1964).

  \bibitem{B-C-R} M.~Berta, M.~Christandl, R.~Renner,
    ``The quantum reverse Shannon theorem based on one-shot information theory'',
    \emph{Commun. Math. Phys.} {\bf 306}, 579 (2011);
    arXiv[quant-ph]:0912.3805.

  \bibitem{Harry} H.~Buhrman, private communication (2015).

  \bibitem{B-C-F-G-G-O-S} H.~Buhrman, N.~Chandran, S.~Fehr, R.~Gelles, V.~Goyal, R.~Ostrovsky, C.~Schaffner,
    ``Position-based quantum cryptography: Impossibility and constructions'';
    arXiv[quant-ph]:1009.2490.

  \bibitem{B-F-S} H.~Buhrman, S.~Fehr, C.~Schaffner,
    ``On the Parallel Repetition of Multi-Player Games: The No-Signaling Case'',
    \emph{Proc. 9th Conference on the Theory of Quantum Computation, Communication
    and Cryptography (TQC14)} {\bf 27}, pp.~24-35 (2014);
    arXiv[quant-ph]:1312.7455.

  \bibitem{C-S} A.~Chailloux, G.~Scarpa,
    ``Parallel repetition of free entangled games: simplification and improvements'';
    arXiv[quant-ph]:1410.4397.

  \bibitem{C-K-R} M.~Christandl, R.~K\"{o}nig, R.~Renner,
    ``Post-selection technique for quantum channels with applications to quantum cryptography'',
    \emph{Phys. Rev. Lett.} {\bf 102}:020504 (2009);
    arXiv[quant-ph]:0809.3019.

  \bibitem{C-W-Y} M.K.~Chung, X.~Wu, H.~Yuen,
    ``Parallel repetition for entangled $k$-player games via fast quantum search'';
    arXiv[quant-ph]:1501.00033 (2015).

  \bibitem{C-H-S-H} J.F.~Clauser, M.A.~Horne, A.~Shimony, R.A.~Holt, ``Proposed experiment to test local hidden-variable theories'', \emph{Phys. Rev. Lett.} {\bf 23}(15), pp.~880-884 (1969).

  \bibitem{C-S-U-U} R.~Cleve, W.~Slofstra, F.~Unger, S.~Upadhyay, ``Strong parallel repetition theorem for quantum XOR proof systems'', \emph{Comput. Complex.} {\bf 17}(2), pp.~282-299 (2008); arXiv:quant-ph/0608146.

  \bibitem{D-S-V} I.~Dinur, D.~Steurer, T.~Vidick, ``A parallel repetition theorem for entangled projection games''; arXiv[quant-ph]:1310.4113.

  \bibitem{D-S-W} R.~Duan, S.~Severini, A.~Winter, ``On zero-error communication via quantum channels in the presence of noiseless feedback''; arXiv[quant-ph]:1502.02987.

  \bibitem{Ebe} P.H.~Eberhard, ``Background level and counter efficiencies required for a loophole-free Einstein-Podolsky-Rosen experiment'', Phys. Rev. A 47, pp.~747-750 (1993).

  \bibitem{F-K-T-W} S.~Fehr, J.~Kaniewski, M.~Tomamichel, S.~Wehner, ``A Monogamy-of-Entanglement Game with Applications to Device-Independent Quantum Cryptography'', \emph{New J. Phys.} {\bf 15}:103002; arXiv[quant-ph]:1210.4359.

  \bibitem{Feige} U.~Feige, ``On the success probability of the two provers in one-
 round proof systems'', \emph{Proc. 6th IEEE Symposium on Structure in Complexity Theory}, pp.~116-123 (1991).

  \bibitem{F-V} U.~Feige, O.~Verbitsky,
    ``Error reduction by parallel repetition -- A negative result'',
    \emph{Combinatorica} {\bf 2}(4), pp.~461-478 (2002).

  \bibitem{F-vdG} C.A.~Fuchs, J.~van de Graaf,
    ``Cryptographic distinguishability measures for quantum-mechanical states'',
    \emph{IEEE Trans. Inf. Theory} {\bf 45}(4), pp.~1216-1227 (1999).

  \bibitem{G-M} A.~Garg, N.D.~Mermin, ``Detector ineffciencies in the Einstein-Podolsky-Rosen experiment'', Phys. Rev. D 35, pp.~3831-3835 (1987).

  \bibitem{Han} E.~H\"{a}nggi, \emph{Device-independent quantum key distribution},
    PhD thesis, ETH Z\"urich (2010);
    arXiv[quant-ph]:1012.3878.



  \bibitem{Hol} T.~Holenstein,
    ``Parallel repetition: simplifications and the no-signaling case'',
    \emph{Theory of Computing} {\bf 5}(1), pp.~141-172 (2009);
    arXiv:cs/0607139.

  \bibitem{Ito} T.~Ito, ``Polynomial-space approximation of no-signaling provers'',
    \emph{Automata, Languages and Programming} {\bf 6198}, pp.~140-151 (2010);
    arXiv[cs.CC]:09082363.

  \bibitem{J-P-Y} R.~Jain, A.~Pereszlenyi, P.~Yao,
    ``A parallel repetition theorem for entangled two-player one-round games
    under product distributions''; arXiv[quant-ph]:1311.6309.


  \bibitem{K-R-R} Y.T.~Kalai, R.~Raz, R.D.~Rothblum, ``How to delegate computations: the power of no-signaling proofs'', \emph{Proc. 46th Annual ACM Symposium on Theory of Computing (STOC14)}, pp.~485-494 (2011).

  \bibitem{K-R} J. Kempe, O. Regev ``No strong parallel repetition with entangled and non-signaling provers'', \emph{Proc. 25th CCC10}, pp.~7-15 (2010); arXiv[quant-ph]:0911.0201.

  \bibitem{K-R-T} J. Kempe, O. Regev, B. Toner, ``Unique games with entangled provers are easy'', \emph{Proc. 49th Annual IEEE Symposium on Foundations of Computer Science (FOCS08)}, pp.~457-466 (2008); arXiv[quant-ph]:0710.0655.

  \bibitem{K-V} J. Kempe, T. Vidick, ``Parallel repetition of entangled games'', \emph{Proc. 43rd Annual ACM Symposium on Theory of Computing (STOC11)}, pp.~353-362 (2011); arXiv[quant-ph]:1012.4728.

  \bibitem{L-W} C. Lancien, A. Winter,
    ``Flexible constrained de Finetti reductions and applications'', in preparation.

  \bibitem{L-GP-R-C} A. Leverrier, R. Garc\'{\i}a-Patr\'{o}n, R. Renner, N.J. Cerf,  ``Security of continuous-variable quantum key distribution against general attacks'', \emph{Phys. Rev. Lett.} {\bf 110}:030502 (2013); arXiv[quant-ph]:1208.4920.

  \bibitem{P-R} S. Popescu, D. Rohrlich, ``Nonlocality as an axiom'', \emph{Foundations of Physics} {\bf 24}(3), pp.~379–385 (1994).

  \bibitem{Rao} A. Rao,
    ``Parallel repetition in projection games and a concentration bound'',
    \emph{SIAM J. Comput.} {\bf 40}(6), pp.~1871-1891 (2011).

  \bibitem{Raz} R. Raz,
    ``A parallel repetition theorem'',
    \emph{SIAM J. Comput.} {\bf 27}(3), pp.~763-803 (1998).

  \bibitem{Ros} R. Rosen,
    ``A $k$-provers parallel repetition theorem for a version of no-signaling model'',
    \emph{Discrete Math., Alg. and Appl.} {\bf 2}(4), pp.~457-468 (2010).

  \bibitem{Tsi} B.S. Tsirelson, ``Quantum generalizations of Bell's inequality'', \emph{Lett. Math. Phys.} {\bf 4}(2), pp.~93-100 (1980).

  \bibitem{Yuen} H. Yuen, ``A parallel repetition theorem for all entangled games''; arXiv[quant-ph]:1604.04340.

\end{thebibliography}
\end{document}